\documentclass{llncs}

\usepackage{latexsym}
\usepackage{amsmath}
\usepackage{amssymb}
\usepackage{mathtools}
\usepackage[table]{xcolor}
\usepackage{graphicx}
\usepackage{verbatim}
\usepackage{xspace}
\usepackage{cite}
\usepackage{enumerate}
\usepackage{nicefrac}
\usepackage{graphicx}
\usepackage{tikz}
\usetikzlibrary{positioning,chains,fit,shapes,calc,arrows}
\usepackage{placeins}
\usepackage{todonotes}

\newcommand{\ab}[1] {\left\vert #1\right\vert}
\DeclarePairedDelimiter{\ceil}{\lceil}{\rceil}

\newcommand{\no}[1] {\left\vert #1\right\vert_1}
\newcommand{\nz}[1] {\left\vert #1\right\vert_0}

\newcommand{\nats}{\ensuremath{\mathbb{N}}\xspace}
\newcommand{\reals}{\ensuremath{\mathbb{R}}\xspace}

\newcommand{\OPT}{\ensuremath{\textsc{Opt}}\xspace}
\newcommand{\ALG}{\ensuremath{\textsc{Alg}}\xspace}
\newcommand{\ALGa}{\ensuremath{\textsc{Alg}_1}\xspace}
\newcommand{\ALGb}{\ensuremath{\textsc{Alg}_2}\xspace}
\newcommand{\inputs}{\ensuremath{\mathcal{I}}\xspace}

\newcommand{\seq}{\ensuremath{\sigma}\xspace}
\newcommand{\length}[1]{\ensuremath{|#1|}\xspace}

\def\opt{\ensuremath{\textsc{Opt}}\xspace}
\def\OPT{\ensuremath{\textsc{Opt}}\xspace}
\def\ALG{\ensuremath{\textsc{Alg}}\xspace}

\def\D{\ensuremath{\textsc{Det}}\xspace}

\def\minasgk{{\sc minASGk}\xspace}
\def\minasgkw{{\sc minASGk$_{\text{w}}$}\xspace}
\def\aoc{\ensuremath{\mathsf{AOC}}\xspace}
\def\aocc{\ensuremath{\mathsf{AOC}\text{-complete}}\xspace}
\def\alg{\ensuremath{\textsc{Alg}}\xspace}

\def\vl{\ensuremath{\mathbf{L}}\xspace}
\def\vv{\ensuremath{\mathbf{v}}\xspace}
\def\vu{\ensuremath{\mathbf{u}}\xspace}

\def\sq{\sqsubseteq}

\def\minSk{{\sc min\-ASGk}\xspace}

\newcommand{\Pw}{\ensuremath{\textsc{P}_{\text{w}}}\xspace}
\def\P{{\textsc P}\xspace}

\DeclareMathOperator{\E}{\mathbb{E}}
\DeclareMathOperator{\polylog}{polylog}

\newcommand{\eps}{\ensuremath{\varepsilon}\xspace}
\newcommand{\wmin}{\ensuremath{w_{\text{min}}}\xspace}
\newcommand{\wmax}{\ensuremath{w_{\text{max}}}\xspace}
\newcommand{\kmax}{\ensuremath{k_{\text{max}}}\xspace}
\newcommand{\vopt}{\ensuremath{V_{\opt}}\xspace}

\newcommand{\imp}{\ensuremath{\text{imp}}\xspace}
\newcommand{\unimp}{\ensuremath{\text{unimp}}\xspace}

\newtheorem{observation}{Observation}

\title{Weighted Online Problems with Advice
\thanks{This work was partially supported by the Villum Foundation, grant VKR023219, and the Danish Council for Independent Research, Natural Sciences, grant DFF-1323-00247.}}
\author{Joan Boyar \and Lene M. Favrholdt \and Christian Kudahl \and Jesper W. Mikkelsen}
\institute{Department of Mathematics and Computer Science,
University of Southern Denmark}

\begin{document}
\maketitle

\begin{abstract}
Recently, the first online complexity class, \aoc, was introduced.
The class consists of many online problems where each request must be either accepted or rejected, and the aim is to either minimize or maximize the number of accepted requests, while maintaining a feasible solution.
All \aoc-complete problems (including Independent Set, Vertex Cover, Dominating Set, and Set Cover) have essentially the same advice complexity.
In this paper, we study weighted versions of problems in \aoc, i.e., each request comes with a weight and the aim is to either minimize or maximize the total weight of the accepted requests.
In contrast to the unweighted versions, we show that there is a significant difference in the advice complexity of complete minimization and maximization problems.
We also show that our algorithmic techniques for dealing with weighted requests can be extended to work for non-complete \aoc 
problems such as Matching in the edge arrival model (giving better results than what follow from the
general \aoc results) and even non-\aoc problems such as scheduling.
\end{abstract}

\section{Introduction}
An online problem is an optimization problem for which the input is divided into small pieces, usually called requests, arriving sequentially.
An online algorithm must serve each request, irrevocably, without any knowledge of possible future requests.
The quality of online algorithms is traditionally measured using the competitive ratio~\cite{CompRatio1, CompRatio2}, which is essentially the worst case ratio of the online performance to the performance of an optimal offline algorithm, i.e., an algorithm that knows the whole input sequence from the beginning and has unlimited computational power.

For some online problems such as Independent Set or Vertex Cover, the best possible competitive ratio is linear in the sequence length.
This gives rise to the question of what would happen, if the algorithm knew {\em something} about future requests.
Semi-online settings, where it is assumed that the algorithm has some specific knowledge such as the value of an optimal solution,  have been studied (see
\cite{BFKLM17} for many relevant references).
The extra knowledge may also be more problem specific such as an access graph for paging~\cite{BIRS95,CN99}.

In contrast to problem specific approaches, advice complexity~\cite{A1,A4,A2} is a quantitative and standardized way of relaxing the online constraint.
The main idea of advice complexity is to provide an online algorithm, \ALG, with some partial knowledge of the future in the form of advice bits provided by a trusted oracle which has unlimited computational power and knows the entire request sequence. 
Informally, the advice complexity of an algorithm
is a function of input sequence length, and for a given $n$, it is the maximum
number of advice bits read for input sequences of length $n$.
The advice complexity of a problem is a function of input sequence length and competitive ratio, and for a given competitive ratio $c$, it is the best possible advice complexity of
any $c$-competitive algorithm for the problem.
Advice complexity is formally defined in Section~\ref{sec:prelim}.

Upper bounds on the advice complexity for a problem can
sometimes lead to (or come from) semi-online algorithms, and lower bounds can show
that such algorithms do not exist. 
Since its introduction, advice complexity has been a very active area of research. Lower and upper bounds on the advice complexity have been obtained for a large number of online problems; a recent list can be found in~\cite{Sml2015}.
For a survey on advice complexity, see~\cite{BFKLM17}.

Recently in~\cite{BFKM15}, the first  complexity class for online problems, 
\aoc, was introduced.
The class consists of online problems that can be described in the following way:
The input is a sequence of requests and each request must either be accepted or rejected.
The set of accepted requests is called the solution.
For each request sequence, there is at least one feasible solution.
The class contains minimization as well as maximization problems.
For a minimization problem, the goal is to accept as few requests as possible, while maintaining a feasible solution, and for maximization problems, the aim is to accept as many requests as possible.
For minimization problems, any super set of a feasible solution is also a solution, and for maximization problems, any subset of a feasible solution is also a feasible solution.
The AOC-complete problems are the hardest problems in the class in terms of
their advice complexity.
The class \aoc is formally defined in Section~\ref{sec:weightedAOC}.

In this paper, we consider a generalization of the problems in the class \aoc in which each request comes with a weight.
The goal is now to either minimize or maximize the total weight of the accepted requests.
We separately consider the classes of maximization and minimization problems.
For \aoc-complete maximization problems, we get advice complexity
results quite similar to those for the unweighted versions of the
problems. On the other hand, for \aoc-complete minimization problems,
the results are a lot more negative: using less than one advice bit per request leads to unbounded competitive ratios, so this gives a complexity class containing harder problems than \aoc.
This is in contrast to unweighted AOC-complete problems, where minimization
and maximization problems are equally hard in terms of advice complexity.
 Recently, differences between (unweighted) AOC minimization and maximization problems were
found with respect to online bounded analysis~\cite{BEFLL16} and
min- and max-induced subgraph problems~\cite{KK15}.

Our upper bound techniques are also useful for non-complete \aoc problems such as Matching in the edge arrival model, as well as non-\aoc problems such as Scheduling.

\paragraph{Previous results.}
For any \aoc-complete problem, $\Theta(n/c)$ advice bits are necessary and sufficient to obtain a competitive ratio of $c$. 
More specifically, for competitive ratio $c$, the advice complexity is $B(n,c) \pm O(\log n)$, where 
\begin{equation}
B(n,c)= \log\left(1+\frac{(c-1)^{c-1}}{c^{c}}\right)n,
\label{eq:bnc}
\end{equation}
and
$an/c \leq B(n,c) \leq n/c$, $a = 1/(e \ln(2)) \approx 0.53$. This is
an upper bound on the advice complexity of all problems in \aoc.
In~\cite{BFKM15}, a list of problems including Independent Set, Vertex Cover,
Dominating Set, and Set Cover were proven \aocc.

The paper~\cite{Amakespan} studies a semi-online version of scheduling where it is allowed to keep several parallel schedules and choose the best schedule in the end.
The scheduling problem considered is makespan minimization on $m$ identical machines.
Using $(1/\eps)^{O(\log(1/\eps))}$ parallel schedules, a $(4/3+\eps)$-competitive algorithm is obtained.
Moreover, a $(1+\eps)$-competitive algorithm which uses $(m/\eps)^{O(\log(1/\eps)/\eps)}$ parallel schedules is given along with an almost matching lower bound.
Note that keeping $s$ different schedules until the end corresponds to working with $s$ different online algorithms.
Thus, this particular semi-online model easily translates to the advice model, the advice being which of the $s$ algorithms to run.
In this way, the results of~\cite{Amakespan} correspond to a $(4/3+\eps)$-competitive algorithm using $O(\log^2(1/\eps))$ advice bits and a $(1+\eps)$-competitive algorithm using $O(\log(m/\eps) \cdot \log(1/\eps) / \eps)$ advice bits.
In particular, note that this algorithm uses constant advice in the size of the input and only logarithmic advice in the number of machines. 

In~\cite{Aschedule}, scheduling on identical machines with a more general type of objective function (including makespan, minimizing the $\ell_p$-norm, and machine covering) was studied. 
The paper considers the advice-with-request model where a fixed number of advice bits are provided along with each request. 
The main result is a $(1+\eps)$-competitive algorithm that uses $O((1/\eps) \cdot \log(1/\eps))$ advice bits per request, totaling $O((n/\eps) \cdot \log (1/\eps))$ bits of advice for the entire sequence.

\paragraph{Our results.}
We prove that adding arbitrary weights, \aoc-complete minimization problems become a lot harder than \aoc-complete maximization problems:
\begin{itemize}
\item For \aocc maximization problems, the weighted version is not significantly harder than the unweighted version:
For any maximization problem in \aoc (this includes, e.g., Independent Set), the $c$-competitive algorithm given in~\cite{BFKM15} for the unweighted version of the problem can be converted into a $(1+\eps)c$-competitive algorithm for the weighted version using only $O((\log^2 n)/\eps)$ additional advice bits.
Thus, a $(1+\eps)c$-competitive algorithm using at most $B(n,c)+O((\log^2 n)/\eps)$ bits of advice is obtained.
For the weighted version of non-complete \aoc maximization problems, a better advice complexity than $B(n,c)$ may be obtained:
For any $c$-competitive algorithm for an \aoc maximization problem, \P, using $b$ advice bits can be converted into a $O(c \cdot \log n)$-competitive algorithm for the weighted version of \P using $b+O(\log n)$ advice bits.
For Weighted Matching in the edge arrival model, this implies a $O(\log n)$-competitive algorithm reading $O(\log n)$ bits of advice.
We show that this is best possible in the following sense:
For a set of weighted \aoc problems including Matching, Independent Set and Clique, no algorithm reading $o(\log n)$ bits of advice can have a competitive ratio bounded by any function of $n$.
Furthermore, any $O(1)$-competitive algorithm for Matching must read $\Omega(n)$ advice bits.
\item For all minimization problems known to be \aocc (this includes, e.g., Vertex Cover, Dominating Set, and Set Cover), $n - O(\log n)$ bits of advice are required to obtain a competitive ratio bounded by a function of $n$.
This should be contrasted with the fact that $n$ bits of advice trivially yields a strictly $1$-competitive algorithm.

If the largest weight \wmax cannot be arbitrarily larger than the smallest weight $\wmin$, the $c$-competitive algorithm given in~\cite{BFKM15} for the unweighted version can be converted into a $c(1+\eps)$-competitive algorithm for the weighted versions using $B(n,c) + O(\log^2 n + \log (\log(\wmax/\wmin)/\eps))$ advice bits in total.
\end{itemize}

Our main upper bound technique is a simple exponential classification scheme that can be used to sparsify the set of possible weights. 
This technique can also be used for problems outside of \aoc.
For example, for {\em scheduling on related machines}, we show that for many important objective functions (including makespan minimization and minimizing the $\ell_p$-norm), there exist $(1+\varepsilon)$-competitive algorithms reading $O((\log^2 n)/\eps)$ bits of advice. 
For scheduling on $m$ {\em unrelated} machines where $m$ is constant, we get a similar result, but with $O((\log n)^{m+1}/\eps^m)$ advice bits.
Finally, for unrelated machines, where the goal is to {\em maximize} an objective function, we show that under some mild assumptions on the objective function (satisfied, for example, for machine covering), there is a $(1+\eps)$-competitive algorithm reading $O((\log n)^{m+1}/\eps^m)$ bits of advice.

For scheduling on related and unrelated machines, our results are the first non-trivial upper bounds on the advice complexity. For the case of makespan minimization on identical machines, the algorithm of~\cite{Amakespan} is strictly better than ours. However, for minimizing the $\ell_p$-norm or maximizing the minimum load on identical machines, we exponentially improve the previous best upper bound~\cite{Aschedule} (which was linear in $n$).

\section{Preliminaries}
\label{sec:prelim}
Throughout the paper, we let $n$ denote the number of requests in the input.
We let $\mathbb{R}_+$ denote the set containing 0 and all positive real numbers.
We let $\log$ denote the binary logarithm $\log_2$.
For $k\geq 1$, $[k]=\{1,2,\ldots, k\}$.
For any bit string $y$, 
let $\nz{y}$ and $\no{y}$ denote the number of zeros and the number of
ones, respectively, in $y$. 
We write $x \sq y$ if for all indices, $i$, $x_i=1 \Rightarrow y_i =1$.

\subsection{Advice complexity and competitive analysis}
\label{subsec:advice}
In this paper, we use the ``advice-on-tape'' model~\cite{A1}.
Before the first request arrives, the oracle, which knows the entire request
sequence, prepares an \emph{advice tape}, an infinite binary string. The algorithm \ALG may, at any point, read some bits from the advice tape. The {\em advice complexity} of \ALG is the maximum number of bits read by \ALG for any input sequence of at most a given length. \OPT is an optimal offline algorithm.

Advice complexity is combined with competitive analysis to determine how many bits of advice are necessary and sufficient to achieve a given competitive ratio.

\begin{definition}[Competitive analysis~\cite{CompRatio1, CompRatio2} and advice complexity~\cite{A1}]
The input to an online problem, \P, is a request sequence $\sigma=\langle r_1,\ldots , r_n \rangle$. An \emph{online algorithm with advice}, \ALG, computes the output $y=\langle y_1,\ldots , y_n\rangle$, where $y_i$ is computed from $\varphi, r_1,\ldots , r_i$, where $\varphi$ is the content of the advice tape. Each possible output for \P is associated with a \emph{cost/profit}. For a request sequence $\sigma$, $\ALG(\sigma)$ $(\OPT(\sigma))$ denotes the 
cost/profit of the output computed by $\ALG$ $(\OPT)$ when serving $\sigma$. 

If \P is a minimization (maximization) problem, then $\ALG$ is \emph{$c(n)$-competitive} if there exists a constant, $\alpha$, such that, for all $n \in \mathbb{N}$, $\ALG(\sigma)\leq c(n)\cdot\OPT (\sigma)+\alpha$,
($\OPT (\sigma)\leq c(n)\cdot\ALG (\sigma)+\alpha$),
 for all request sequences, $\sigma$, of length at most $n$. 
If the relevant inequality holds with $\alpha=0$, we say that $\ALG$ is \emph{strictly $c(n)$-competitive}.

The \emph{advice complexity, $b(n)$, of an algorithm}, \ALG, is the largest number of bits of $\varphi$ read by \ALG over all possible request sequences of length at most $n$. 
The {\em advice complexity of a problem}, \P, is a function, $f(n,c)$, $c\geq 1$, such that the smallest possible advice complexity of a strictly $c$-competitive online algorithm for \P is $f(n,c)$.
\end{definition}

We only consider deterministic online algorithms (with advice). Note that both 
$b(n)$ and $c(n)$ in the above definition may depend on $n$, but, for ease of notation, we often write $b$ and $c$ instead of $b(n)$ and $c(n)$. Also, with this definition, $c\geq 1$, for both minimization and maximization problems. 

\subsection{Complexity classes}
In this paper, we consider the complexity class \aoc from~\cite{BFKM15}.

\begin{definition}[\aoc~\cite{BFKM15}]\label{sgeasydef}
A problem, \P, is in \aoc \emph{(Asymmetric Online Covering)} if it can be defined
 as follows:
The input to an instance of \P consists of a sequence of $n$ requests, $\sigma= 
\langle r_1, \ldots, r_n \rangle$, and possibly one final dummy request. An algorithm for \P computes a binary output string, $y=y_1 \ldots y_n\in\{0,1\}^n$, where $y_i=f(r_1, \ldots , r_i)$ for some function $f$. 

For minimization (maximization) problems, the score function, $s$, maps a pair, 
$(\sigma,y)$, of input and output to a cost (profit) in $\mathbb{N} \cup \{\infty \}$ $(\mathbb{N} \cup \{-\infty \})$.
 For an input, $\sigma$, and an output, $y$, $y$ is \emph{feasible} if $s(\sigma
,y) \in \mathbb{N}$. Otherwise, $y$ is \emph{infeasible}. There must exist at least one feasible output. Let $S_{\min}(\sigma)$ $(S_{\max}(\sigma))$ be the set 
of those outputs that minimize (maximize) $s$ for a given input $\sigma$.

If \P is a minimization problem, then for every input, $\sigma$, the following must hold:
\begin{enumerate}
\item For a feasible output, $y$, $s(\sigma,y)=\no{y}$.
\item An output, $y$, is feasible if 
   there exists a $y'\in S_{\min}(\sigma)$ such that $y'\sq y$.\\
  If there is no such $y'$, the output may or may not be feasible.
\end{enumerate}

If \P is a maximization problem, then for every input, $\sigma$, the following must hold:
\begin{enumerate}
\item For a feasible output, $y$, $s(\sigma,y)=\nz{y}$.
\item An output, $y$, is feasible if 
   there exists a $y'\in S_{\max}(\sigma)$ such that $y'\sq y$.\\
  If there is no such $y'$, the output may or may not be feasible.
\end{enumerate}
\end{definition}

Recall that no problem in \aoc requires more than $B(n,c) + O(\log n)$ bits 
of advice (see Eq.~(\ref{eq:bnc}) for the definition of $B(n,c)$). 
This result is based on a covering design technique, where the advice
indicates a superset of the output bits that are 1 in an optimal
solution.

The problems in \aoc requiring the most advice are \aocc~\cite{BFKM15}:
\begin{definition}[\aocc~\cite{BFKM15}]
\label{completedef}
A problem $\P \in \aoc$ is \emph{\aocc} if for all $c>1$, any $c$-competitive algorithm for \P must read at least $B(n,c)-O(\log n)$
  bits of advice.
\end{definition}

In~\cite{BFKM15}, an abstract guessing game, \minSk (Minimum Asymmetric String Guessing with Known History), was introduced and shown to be \aocc. The \minSk-problem itself is very artificial, but it is well-suited as the starting point of reductions. All minimization problems known to be \aocc have been shown to be so via reductions from \minSk. 

The input for \minSk is a secret string $x=x_1x_2\ldots x_n\in\{0,1\}^n$ given in $n$ rounds. In round $i\in[n]$, the online algorithm must answer $y_i\in\{0,1\}$. Immediately after answering, the correct answer $x_i$ for round $i$ is revealed to the algorithm. If the algorithm answers $y_i=1$, it incurs a cost of $1$. If the algorithm answers $y_i=0$, then it incurs no cost if $x_i=0$, but if $x_i=1$, then the output of the algorithm is declared to be infeasible (and the algorithm incurs a cost of $\infty$). The objective is to minimize the total cost incurred. Note that the optimal solution has cost $\no{x}$. See the appendix
for a formal definition of \minSk and for definitions of other \aocc problems.

The problem \minSk is based on the {\em
  binary string guessing} problem~\cite{SG,A2}.
Binary string guessing is similar to asymmetric string guessing,
except that any wrong guess (0 instead of 1 or 1 instead of 0) gives a
cost of 1.

In Theorem \ref{minasglower}, we show a very strong lower bound for a weighted version of \minSk. In Theorem \ref{thm:reduction}, via reductions, we show that this lower bound implies similar strong lower bounds for the weighted version of other \aocc minimization problems.

\begin{definition}[Weighted \aoc]
Let $\P$ be a problem in \aoc. We define the {\em weighted version of $\P$}, denoted $\Pw$, as follows: 
A $\Pw$-input $\sigma=\langle \{r_1, w_1\},$ $\{r_2, w_2\},$ $\ldots , \{r_n, w_n\}\rangle$ consists of $n$ $\P$-requests,
 $r_1,...,r_n$, each of which has a weight $w_i\in\mathbb{R}_+$. The $\P$-request $r_i$ and its weight $w_i$ are revealed simultaneously.
 An output $y=y_1\ldots y_n\in\{0,1\}^n$ is feasible for the input $\sigma$
 if and only if $y$ is feasible for the $\P$-input $\langle r_1,\ldots , r_n\rangle$.
 The cost (profit) of an infeasible solution is $\infty$ ($-\infty$).

If $\P$ is a minimization problem, then the cost of a feasible $\Pw$-output $y$ for an input $\sigma$ is
$$
s(\sigma,y)=\sum_{i=1}^n w_iy_i
$$

If $\P$ is a maximization problem, then the profit of a feasible $\Pw$-output $y$ for an input $\sigma$ is
$$
s(\sigma,y)=\sum_{i=1}^n w_i (1-y_i)
$$
\end{definition}

\section{Weighted Versions of \aoc-Complete Minimization Problems}
\label{sec:weightedAOC}

In the weighted version of \minSk, \minasgkw, each request is a weight for the
current request
and the value $0$ or $1$ of the previous request. Producing a feasible solution
requires {\em accepting} (answering $1$ to) all requests with value $1$,
and the cost of a feasible solution is the sum of all weights for requests 
which are accepted.

We start with a negative result for \minasgkw and then use it to
obtain similar results for the weighted online version of Vertex
Cover, Set Cover, Dominating Set, and Cycle Finding.

\begin{theorem}
\label{minasglower}
For \minasgkw, no algorithm using less than $n$ bits of advice is
$f(n)$-competitive, for any function $f$.
\end{theorem}
\begin{proof}
Let \alg be any algorithm for \minasgkw reading at most $n-1$ bits of
advice.
We show how an adversary can construct input sequences where the cost
of \alg is arbitrarily larger than that of \opt.
We only consider sequences with at least one 1.
It is easy to see that for the unweighted version of the binary string
guessing problem, $n$ bits of advice are necessary in order to guess
correctly each time: If there are fewer than $n$ bits, there are only
$2^{n-1}$ possible advice strings, so, even if we only consider the
$2^n-1$ possible inputs
with at least one 1, there are at least two different request
strings, $x$ and $y$, which get the same advice string. \alg
will make an error on one of the strings when guessing the first bit
where $x$ and $y$ differ, since up until that point \alg
has the same information about both strings.

We describe a way to assign weights to the requests in \minasgkw such that
if \alg makes a single mistake (either guessing 0 when the correct answer is 1
or vice versa), its performance ratio is unbounded.
We  use a large number $a>1$, which
we allow to depend on $n$. 
All weights are from the interval $[1,a]$ (note that they are not necessarily integers).
We let $x=x_1, \ldots, x_n$ be the input string and set $w_1=a^{1/2}$. For $i>1$,
$w_i$ is given by:
\[
 w_i =
  \begin{cases} 
      \hfill w_{i-1}\cdot a^{(-2^{-i})},    \hfill & \text{ if $x_{i-1}=0$} \\ 
      \hfill w_{i-1}\cdot a^{(2^{-i})}, \hfill & \text{ if $x_{i-1}=1$} \\
  \end{cases}
\]
Since the weights are only a function of previous requests, they do not reveal any information to \alg about future requests.

\begin{observation} \label{min_opt_obs}
For each $i$, the following hold:
\begin{enumerate}[(a)]
\item \label{obs:zero}
      If $x_i=0$, then $w_j \leq w_{i}\cdot a^{(-2^{-n})}$ for all $j>i$.
\item \label{obs:one}
      If $x_i=1$, then $w_j \geq w_{i}\cdot a^{(2^{-n})}$ for all $j>i$.
\end{enumerate}
\end{observation}
We argue for each set of inequalities in the observation:

(\ref{obs:zero}): If $x_i=0$, for each $j>i$, $w_j=w_i\cdot a^{(-2^{-(i+1)})}\cdot a^{\sum_{k=i+2}^j
(\pm 2^{-(i+1)})}$, where the plus or minus depends on whether $x_k=0$
or $x_k=1$. The value $w_j$ is largest if all of the $x_k$ values are $1$,
in which case $w_j=w_i\cdot a^{(-2^{-j})} \leq w_i\cdot
a^{(-2^{-n})}$. 

(\ref{obs:one}): The
argument of $x_i=1$ is similar, changing minus to plus and vice versa.

We claim that if \alg makes a single mistake, its performance ratio is not bounded by any function 
of $n$. 
Indeed, if \alg guesses $0$ for a request, but the correct answer is $1$, the solution is
infeasible and \alg gets a cost of $\infty$.

We now consider the case where \alg guesses $1$ for a request $j$, 
but the correct answer is $0$.
This request gives a contribution of  $w_j =a^b$, for some $0<b<1$, 
to the cost of the solution produced by \alg.
Define $j'$ such that $w_{j'}=\max \{w_i \mid x_i=1\}$. 
Since \opt only answers $1$ if $x_i =1$, this is the largest contribution 
to the cost of \opt from a single request.

If $j'>j$, Observation~\ref{min_opt_obs}(\ref{obs:zero})
gives that $w_{j'} \leq w_j\cdot a^{(-2^{-n})}=a^b\cdot a^{(-2^{-n})}=a^{b-2^{-n}}$.
The cost of \opt is at most $n \cdot w_{j'} \leq n \cdot a^{b-2^{-n}}$.
Thus, 
\[
    \frac{\alg(x)}{\opt(x)} \geq \frac{a^b}{n \cdot a^{b-2^{-n}}}=\frac{a^{2^{-n}}}{n}.
\]
Since $a$ can be arbitrarily large (recall that it can be a function of $n$),
we see that no algorithm can be $f(n)$-competitive for any specific function $f$.

If $j'<j$, Observation~\ref{min_opt_obs}(\ref{obs:one}) gives us that 
$w_{j} \geq w_{j'}\cdot a^{(2^{-n})}$. Using $w_j=a^b$, we get
$a^{b-2^{-n}}  \geq w_{j'} $. We can repeat the argument from the case where
$j'>j$ to see that  no algorithm can be $f(n)$-competitive for any specific function $f$.
\qed
\end{proof}

In order to show that similar lower bounds apply to all minimization
problems known to be complete for \aoc,
we define a simple type of advice preserving reduction 
for 
online problems.
 These are much
less general than those defined by Sprock
in his PhD dissertation~\cite{Sprock13}, mainly because we do not allow the amount of
advice needed to change by a multiplicative factor.

Let $\OPT_{\P}(\seq)$ denote the value of the optimal solution for
request sequence $\seq$ for problem $\P$,
and let $|\seq|$ denote the number of requests
in  $\seq$.

\begin{definition}
\label{def:reduction}
Let $\P_1$ and $\P_2$ be two online minimization problems, and let $\inputs_1$ be the
set of request sequences for $\P_1$ and $\inputs_2$ be the set of request
sequences for $\P_2$. 
For a given function $g: \nats \rightarrow \reals_+$, we say that there is a
{\em length preserving $g$-reduction} from $\P_1$ to $\P_2$, if there is
a {\em transformation
function} $f: \inputs_1 \rightarrow \inputs_2$ such that 
\begin{itemize}
\item for all $\seq\in \inputs_1$, $\length{\seq} = \length{f(\seq)}$, and
\item
for every algorithm \ALGb for $\P_2$,
there is an algorithm \ALGa for $\P_1$ such that for all $\seq_1\in
\inputs_1$, the following holds:\\
If \ALGb produces a feasible solution for $\seq_2 = f(\seq_1)$ with
advice $\phi(\seq_2)$, then \ALGa, using at most $\length{\phi(\seq_2)} + g(\length{\seq_2})$ advice bits, produces a
feasible solution for $\seq_1$ such that 
\begin{itemize}
\item 
$\ALGa(\seq_1) \leq \ALGb(\seq_2) + \OPT_{\P_1}(\seq_1)$ and $\OPT_{\P_1}(\seq_1)
\geq \OPT_{\P_2}(\seq_2)$, or
\item $\ALGa(\seq_1) = \OPT_{\P_1}(\seq_1)$
\end{itemize}
\end{itemize}
\end{definition}

Note that the transformation function $f$ is length-preserving in that
the lengths of the request sequences for the two problems are identical. 
This avoids the
potential problem that the advice for the two problems could be functions
of two different sequence lengths. The amount of advice for the
problem being reduced to is allowed to be an additive function,
$g(n)$, longer than for the original problem, because this seems
to be necessary for some of the reductions showing that problems
are \aocc. Since the reductions are only used here to show that
no algorithm is $F(n)$-competitive for any function $F$,
the increase in the performance
ratio that occurs with these reductions is insignificant. 

The following lemma shows how length-preserving reductions can be used.

\begin{lemma}
\label{lemma:reduction}
Let $\P_1$ and $\P_2$ be online minimization problems.
Suppose that at least $b_1(n,c)$ advice bits are required 
to be $(c+1)$-competitive
for $\P_1$ and suppose there is a length preserving
$g(n)$-reduction from $\P_1$ to $\P_2$.
Then, at least $b_1(n,c)-g(n)$ advice bits are needed for an
algorithm for $\P_2$ to be $c$-competitive.
\end{lemma}
\begin{proof}
Let $f$ be the transformation function associated with $g$.
Suppose for the sake of contradiction that there is a (strictly) $c$-competitive algorithm $\ALGb$ for $\P_2$ with advice complexity
$b_2(n,c) < b_1(n,c)-g(n)$. Then there exists a constant $\alpha$ such that for any request
sequence $\seq_1 \in \inputs_1$, either $\ALGa(\seq_1) =\OPT_{\P_1}(\seq_1)$ or
\begin{align*}
 \ALGa(\seq_1)
& \leq \ALGb(\seq_2) + \OPT_{\P_1}(\seq_1)\\
& \leq c\cdot\OPT_{\P_2}(\seq_2) + \alpha + \OPT_{\P_1}(\seq_1)\\
& \leq (c+1)\cdot\OPT_{\P_1}(\seq_1) + \alpha,
\end{align*}
where $\seq_2 = f(\seq_1)$.
Thus, \ALGa is (strictly) $(c+1)$-competitive, with less than
$b_1(n,c)$ bits of advice, a contradiction.
\qed
\end{proof}

All known \aocc problems were proven complete using length-preserving
reductions from \minasgk, so the following holds for the weighted versions of 
all such problems:

\begin{theorem}
\label{thm:reduction}
For the weighted online versions of Vertex Cover, Cycle Finding, Dominating Set, Set Cover, an algorithm reading
less than $n-O(\log n)$ bits of advice cannot be $f(n)$-competitive for
any function $f$.
\end{theorem}

\begin{proof}
The reductions in~\cite{BFKM15} showing that these problems are \aocc
are length preserving $O(\log n)$-reductions from \minasgk, and hence,
the theorem follows from Lemma~\ref{lemma:reduction}. 
For Vertex Cover, the following $O(\log n)$-reduction can be used
(the other three reductions are given in the Appendix~\ref{reductions}):

Each input $\seq = \langle x_1, x_2, \ldots, x_n \rangle$ to the problem
\minasgk, is transformed to $f(\seq) = \langle v_1, v_2, \ldots, v_n \rangle$, where
$V=\{v_1, v_2, \ldots, v_n \}$ is the vertex set of a graph with edge
set $E=\{(v_i,v_j) \colon x_i=1 \text{ and } i<j\}.$
Let $V_1 = \{v_i \in V \colon x_1 = 1\}$.
Note that $V_1 \setminus \{v_n\}$ is a minimum vertex cover of the
graph and that no algorithm can reject more than one vertex from
$V_1$, since $V_1$ induces a clique.

The advice used by the \minasgk algorithm \ALGa consists of the advice
used by the Vertex Cover algorithm \ALGb and $O(\log n)$ bits that are
either all 0 or give (an encoding of)
an index to a position $i$ in the input sequence, such that $v_i \in
V_1$ and \ALGb rejects $v_i$.

Let $V_{\ALGb} \subseteq V$ be the vertex cover constructed by \ALGb and
let $X_{\ALGa}$ be the set of requests on which \ALGa returns a 1.
Then either $X_{\ALGa} = V_{\ALGa}$ or $X_{\ALGa} = V_{\ALGb} \cup \{v_i\}$,
where $\{v_i\} = V_1 \setminus V_{\ALGa}$.
Thus,
$\ALGa(\seq) \leq \ALGb(f(\seq))+\OPT(\seq)$, since $w_i \leq \opt(\seq)$.
\qed
\end{proof}

\section{Exponential Sparsification}
Assume that we are faced with an online problem for which we know how to obtain a
reasonable competitive ratio, possibly using advice, in the unweighted version (or when there are only few possible different weights). We use \emph{exponential sparsification}, a simple technique which can be of help when designing algorithms with advice for weighted online problems by reducing the number of different possible weights the algorithm has to handle. The first step is to partition the set of possible weights into intervals of exponentially increasing length, i.e., for some small $\varepsilon$, $0<\varepsilon <1$,
$$\mathbb{R}_+=\bigcup_{k=-\infty}^{\infty}\big[(1+\varepsilon)^k, (1+\varepsilon)^{k+1}\big).$$
How to proceed depends on the problem at hand. We now informally explain the meta-algorithm that we repeatedly use in this paper. Note that if $w_1,w_2\in \big[(1+\varepsilon)^k, (1+\varepsilon)^{k+1}\big)$ and $w_1\leq w_2$, then $w_1\leq w_2\leq (1+\varepsilon)w_1$. For many online problems, this means that an algorithm can treat all requests whose weights belong to this interval as if they all had weight $(1+\varepsilon)^{k+1}$ with only a small loss in competitiveness. 

Consider now a set of weights and let $\wmax$ denote the largest
weight in the set.
Let \kmax be the integer for which $\wmax\in \big[(1+\varepsilon)^{\kmax},
  (1+\varepsilon)^{\kmax+1}\big)$.
We say that a request with weight $w\in \big[(1+\varepsilon)^{k},
  (1+\varepsilon)^{k+1}\big)$ is \emph{unimportant} if
$k<\kmax- \ceil{\log_{1+\varepsilon}(n^2)}$.
Furthermore, we will often categorize the request as \emph{important} if
$\kmax-\ceil{\log_{1+\varepsilon}(n^2)} \leq k < \kmax+1$ and as {\em
  huge} if $k \geq \kmax+1$.
Each unimportant request has weight $w \leq (1+\varepsilon)^{k+1}\leq 
(1+\varepsilon)^{\kmax-\ceil{\log_{1+\varepsilon}(n^2)}-1+1}\leq 
\wmax/n^2$, so the total sum of the unimportant weights is $O(\wmax/n)$.
 For many weighted online problems, this means that an algorithm can easily serve the requests with unimportant weights, as follows.
In maximization problems, this is done by rejecting them. In minimization problems, it is done by accepting them.
 Thus, exponential sparsification (when applicable) essentially reduces the problem of computing a good approximate solution for a problem with $n$ distinct weights to that of computing a good approximate solution with only $O(\log_{1+\varepsilon} n)$ distinct weights. 

For a concrete problem, several modifications of this meta-algorithm might be necessary. Often, the most tricky part is how the algorithm can learn $\kmax$ without using too much advice. One approach that we often use is the following: The oracle encodes the index $i$ of the first request whose weight is close enough to $(1+\varepsilon)^{\kmax}$ that the algorithm only needs a little bit of advice to deduce $\kmax$ from the weight of this request. If it is somehow possible for the algorithm to serve all requests prior to $i$ reasonably well, then this approach works well.

Our main application of exponential sparsification is to weighted \aoc problems. We begin by considering maximization problems.
Note that no assumptions are made about the weights of $\Pw$ in Theorem~\ref{exsparse}.

\begin{theorem}
\label{exsparse}
If $\P\in\aoc$ is a maximization problem, then for any $c>1$ and $0<\varepsilon\leq 1$,
\Pw has a strictly $(1+\varepsilon)c$-competitive algorithm using $B(n,c)+O(\varepsilon^{-1}\log^2 n)$ advice bits.
 \end{theorem}
\begin{proof}
Fix $\varepsilon>0$. Let $\sigma=\langle \{r_1,w_1\},\ldots ,\{r_n, w_n\}\rangle$ be the input and let
$x=x_1\ldots x_n\in\{0,1\}^n$ specify an optimal solution for $\sigma$, with zeros indicating membership in the optimal solution. Throughout most of this
proof, we assume that $n$ is sufficiently large. The necessary conditions
are discussed at the end of the proof, along with how to handle small $n$.

Define $s=1+\varepsilon/2$. Let $\vopt=\{i\colon x_i=0\}$. Note that $\vopt$ 
contains exactly those rounds in which $\OPT$ answers $0$ and thus accepts. 
Furthermore, for $k\in\mathbb{Z}$, let $V^k=\{i\colon s^k\leq w(i) <s^{k+1}\}$ and let $\vopt^k=\vopt\cap V^k$. Finally, let $i_{\max}\in \vopt$ be such that $w(i_{\max})\geq w(i)$ for every $i\in \vopt$.

The oracle computes the unique $m\in\mathbb{Z}$ such that $i_{\max}\in \vopt^m$. 
We say that a request $r_i$ is \emph{unimportant} if $w(i)<s^{m- \ceil{\log_{s}(n^2)}}$, \emph{important} if $s^{m-\ceil{\log_{s}(n^2)}} \leq w(i) < s^{m+1}$, and \emph{huge} if $w(i)\geq s^{m+1}$. The oracle computes the index $i'$ of the first important request in the input sequence. Assume that $i' \in V^{m'}$. The oracle writes the length $n$ of the input onto the advice tape using a self-delimiting encoding\footnote{For example, $\lceil \log n\rceil$ could be written in unary ($\lceil \log n\rceil$ ones, followed by a zero) before writing $n$ itself in binary.}, and then writes the index $i'$ and the integer $m-m'$ (which is at most $\ceil{\log_{s}(n^2)}$) onto the tape, using a total of $O(\log n)$ bits.
This advice allows the algorithm to learn $m$ as soon as the first important request arrives. 
From there on, the algorithm will know if a request is important, unimportant, or huge. 
Whenever an unimportant or a huge request arrives, the algorithm answers $1$ (rejects the request). 
We now describe how the algorithm and oracle work for the important requests.

For each $0\leq j\leq \ceil{\log_{s}(n^2)}$, let $n_{m-j}=\ab{V^{m-j}}$.
For the requests (whose indices are) in $V^{m-j}$, we use the covering design based $c$-competitive algorithm for unweighted \aoc-problems. 
This requires $B(n_{m-j}, c) + O(\log n_{m-j})$ bits of advice. 
Since $B(n,c)$ is linear in $n$, this means that we  use a total of 
$$b=\sum_{j=0}^{\ceil{\log_{s}(n^2)}} \big(B(n_{m-j}, c) + O(\log n_{m-j})\big)\leq B(n,c)+O(\log_{s} n \cdot\log n)$$ 
bits of advice. 
Note that $\log_{s}(n)\leq 2\varepsilon^{-1}\log n$ for $\varepsilon/2\leq 1$, giving the bound
on the advice in the statement of the theorem.

We now prove that the algorithm achieves the desired competitiveness. 
We can ignore the huge requests, since neither $\ALG$ nor $\OPT$ accepts
any of them.
Let $V_{\ALG}$ be those rounds in which $\ALG$ answers $0$ and let $V_{\ALG}^k=V_{\ALG}\cap V^k$. 
We consider the important requests first. Fix $0\leq j\leq \ceil{\log_{s}(n^2)}$. 
Let $n^{\OPT}_{m-j}=\ab{\vopt^{m-j}}$, i.e., $n^{\OPT}_{m-j}$ is the number of requests in $V^{m-j}$ which are also in the optimal solution $\vopt$.
 By construction, we have $n^{\OPT}_{m-j}\leq c\ab{V_{\ALG}^{m-j}}$.
 Since the largest possible weight of a request in $V^{m-j}$ is at most $s$ times larger than the smallest possible weight of a request in $V^{m-j}$, 
 this implies that $w(\vopt^{m-j})\leq s \cdot c\cdot w(V_{\ALG}^{m-j})$. Thus, we get that
\begin{equation}
\label{eq:uppermax1}
\sum_{j=0}^{\ceil{\log_{s}(n^2)}}w(\vopt^{m-j})\leq
\sum_{j=0}^{\ceil{\log_{s}(n^2)}}s \cdot c\cdot w(V_{\ALG}^{m-j})= s\cdot c\cdot \ALG(\sigma).
\end{equation} 

We now consider the unimportant requests. If $r_i$ is unimportant, then $w(i)\leq s^{m'-\ceil{\log_{s}(n^2)}}\leq s^{m'} / n^2 \leq w(x_{i_{\max}})/n^2\leq \OPT(\sigma)/n^2$. This implies that
\begin{equation}
\label{eq:uppermax2}
\sum_{j=\ceil{\log_{s}(n^2)}+1}^{\infty}w(V_{\OPT}^{m-j})\leq n\frac{\OPT(\sigma)}{n^2}=\frac{\OPT(\sigma)}{n}.
\end{equation}
We conclude that
$$\OPT(\sigma)=w(\vopt)=\sum_{j=0}^{\ceil{\log_{s}(n^2)}}w(\vopt^{m-j})+\sum_{j=\ceil{\log_{s}(n^2)}+1}^{\infty}w(\vopt^{m-j}).$$
By Eq.~(\ref{eq:uppermax2}),
$$\left( 1-\frac{1}{n}\right)\OPT(\sigma) \leq \sum_{j=0}^{\ceil{\log_{s}(n^2)}}w(\vopt^{m-j}),$$
so by Eq.~(\ref{eq:uppermax1}),
$\OPT(\sigma) \leq \frac{n}{n-1}\cdot s\cdot c\cdot \ALG(\sigma)$.

Note that  for $n\geq n_0=\frac{2+2\varepsilon}{\varepsilon}$, 
$(\frac{n}{n-1})(1+\varepsilon/2)\leq (1+\varepsilon)$. For inputs of length less than $n_0$, the oracle writes an optimal solution onto the advice tape, using at
most $n_0$ bits. Since $n_0\leq \frac{4}{\varepsilon}$, $b\in O(\varepsilon^{-1}\log^2 n)$ as required. For inputs of length at least $n_0$, we use the algorithm described above. Thus, for every input $\sigma$, it holds that $\OPT(\sigma)\leq (1+\varepsilon)c\ALG(\sigma)$. Since $\varepsilon$ was arbitrary, this proves the theorem.
\qed
\end{proof}
It may be surprising that adding weights to \aocc maximization problems has almost no effect, while adding weights to \aocc minimization problems drastically changes the advice complexity.
In particular, one might wonder why the technique used in Theorem~\ref{exsparse} does not work for minimization problems. 
The key difference lies in the beginning of the sequence. Let
$w_{\max}$ be the largest weight of a request accepted by \OPT. 

For maximization problems, the algorithm can safely reject all
requests before the first important one.
For minimization problems, this approach does not work, since the algorithm must accept a superset of what \opt accepts in order to ensure that its output is feasible. Thus, rejecting an unimportant request that \opt accepts may result in an infeasible solution. This essentially means that the algorithm is forced into accepting all requests before the first important request arrives.  Accepting all unimportant
requests is no problem, since they will not contribute significantly to the total cost. However, accepting even a single huge request can give an unbounded contribution to the algorithm's cost. As shown in Theorem~\ref{minasglower}, it is not possible in general for the algorithm to tell if a request in the beginning of the sequence is unimportant or huge without using a lot of advice.

However, if the ratio of the largest to the smallest weight is not
too large, exponential sparsification is also useful for minimization problems in
\aoc. Essentially, when this ratio is bounded, it is possible for the algorithm to learn a good approximation of $w_{\max}$  when the first request arrives. This is formalized in Theorem~\ref{exsparsemin}, the proof of which is very similar to the proof of Theorem~\ref{exsparse}.

\begin{theorem}
\label{exsparsemin}
If $\P\in\aoc$ is a minimization problem and $0<\varepsilon\leq 1$, then \Pw with all weights in $[\wmin,\wmax]$
has a $(1+\varepsilon)c$-competitive algorithm with advice complexity at most
$$B(n,c)+O\left(\varepsilon^{-1} \log^2 n+ \log \left( \eps^{-1} \log \frac{\wmax}{\wmin}\right)\right)\,.$$ 
\end{theorem}
\begin{proof}
Fix $\varepsilon>0$. Let $\sigma=\langle \{r_1,w_1\},\ldots ,\{r_n, w_n\}\rangle$ be the input and let
$x=x_1\ldots x_n\in\{0,1\}^n$ specify an optimal solution for $\sigma$, with ones indicating membership in the optimal solution. 
Define $s=1+\varepsilon/2$. 
Let $\vopt=\{i\colon x_i=1\}$. Note that $\vopt$ 
contains exactly those rounds in which $\OPT$ answers $1$ and thus accepts. 
Furthermore, for $k\in\mathbb{Z}$, let $V^k=\{i\colon s^k\leq w(i) <s^{k+1}\}$
 and let $\vopt^k=\vopt\cap V^k$. 
 Finally, let $i_{\max}\in \vopt$ be such that $w(i_{\max})\geq w(i)$ for every $i\in \vopt$.

The oracle computes the unique $m\in\mathbb{Z}$ such that $i_{\max}\in \vopt^m$. 
We say that a request $r_i$ is \emph{unimportant} if $w(i)<s^{m- \ceil{\log_{s}(n^2)}}$, \emph{important} if $s^{m-\ceil{\log_{s}(n^2)}} \leq w(i) < s^{m+1}$, and
 \emph{huge} if $w(i) \geq s^{m+1}$. 
The oracle also computes the unique $m\in\mathbb{Z}$ such that $s^{m'}
\leq w_1 < s^{m'+1}$ and writes the values $n$ and $m-m'$ on the tape
in a self-delimiting encoding.

The number of advice bits needed to write $m-m'$ is $O(\log (m-m'))$.
\begin{align*}
\log{(m-m')} 
& \leq \log{(\log_s{w(i_{\max})} - \log_s{w_1})}+1 \\
& \leq \log{\log_s{\frac{\wmax}{\wmin}}}+1\\
& \leq \log \left( 2\eps^{-1} \log \frac{\wmax}{\wmin} \right)+1, \text{ since } \log_s n \leq 2 \eps^{-1}\log n, \text{ for } \eps/2 \leq 1
\end{align*}
Note that since the length of $m-m'$ is not known, we need to use a self-delimiting encoding, which means that we use $O(\log n+\log ( \eps^{-1} \log \frac{\wmax}{\wmin}))$
advice bits at the beginning.

This advice allows the algorithm to learn $m$ as soon as the first
request arrives.
From there on, the algorithm will
 know if a request is important, unimportant, or huge. 
Whenever a huge request arrives, the algorithm answers $0$ (rejects the request).
When an unimportant request arrives, the algorithm answers $1$ (accepts the request).  
We now describe how the algorithm and oracle work for the important requests.

For the important requests (whose indices are) in $V^{m-j}$, we use the covering design based $c$-competitive algorithm for unweighted \aoc-problems. 
This is similar to what we do in the proof of Theorem~\ref{exsparse}.
The same calculations yield an upper bound on this advice of $B(n,c)+O(\log_{s} n \cdot\log n)$.
Note that $\log_{s}(n)\leq 2\varepsilon^{-1}\log n$ for $\varepsilon/2\leq 1$, giving the bound
on the advice in the statement of the theorem.

First, we note that the solution produced is valid, since it is a superset of the solution of \opt.

We now argue that the cost of the solution is at most $(1+\varepsilon)c$ times the cost of \opt.
Following the proof of Theorem~\ref{exsparse} and switching the roles of \opt
and \alg, we have by construction that the cost of the important requests for the algorithm
is at most $sc$ times larger than the cost for \opt on the important requests. For the huge requests, both this
algorithm and \opt incur a cost of zero.

We now consider the unimportant requests. 
If $r_i$ is unimportant, then $$w(i)< s^{m-\ceil{\log_{s}(n^2)}}\leq s^{m} / n^2 \leq w(x_{i_{\max}})/n^2\leq \OPT(\sigma)/n^2.$$ 
This implies that
\begin{equation}
\sum_{j=\ceil{\log_{s}(n^2)}+1}^{\infty}w(V^{m-j})\leq n\frac{\OPT(\sigma)}{n^2}=\frac{\OPT(\sigma)}{n}.
\end{equation}

Thus, even if the algorithm accepts all unimportant requests and \opt accepts none of them,
it only accepts an additional $\frac{\OPT(\sigma)}{n}$.
In total, the algorithm gets a cost of at most $(1+\frac{1}{n})(1+\varepsilon/2)c\OPT(\sigma)$.
For $n\geq n_0=\frac{2+\varepsilon}{\varepsilon}$, this is at most $(1+\varepsilon)c\OPT(\sigma)$.
For inputs of length less than $n_0$, the oracle will write an optimal solution onto the advice tape, using at
most $n_0$ bits. Since $n_0\leq \frac{3}{\varepsilon}$, $b\in O(\varepsilon^{-1}\log^2 n)$ as required. For inputs of length at least $n_0$, we use the algorithm described above. 
Thus, for every input $\sigma$, it holds that $\OPT(\sigma)\leq (1+\varepsilon)c\ALG(\sigma)$. 
\qed
\end{proof}

\section{Matching and Other Non-Complete \aoc Problems}
We first provide a general theorem that works for all maximization problems
 in \aoc, giving better results in some cases than that in Theorem~\ref{exsparse}. 
\begin{theorem}
\label{thm:WeightedAOCgeneral}
Let $\P \in \aoc$ be a maximization problem.
If there exists a $c$-competitive \P-algorithm reading $b$ bits of advice, then there exists a $O(c\cdot \log n)$-competitive \Pw-algorithm reading $O(b+\log n)$ bits of advice.
\end{theorem}
\begin{proof}
Use exponential sparsification on the weights with an arbitrary $\varepsilon$,
say $\varepsilon=1/2$, and let $s=1+\varepsilon$. For a given request sequence,
$\sigma$, let $w_{\max}$ be the 
maximum weight that $\opt_{\Pw}$ accepts. 
The oracle computes the unique $m \in \mathbb{Z}$ such that $w_{\max} \in [s^m,s^{m+1})$.
The important requests are those with weight $w$, where $s^{m-\ceil{\log_s (n^2)}} \leq w < s^{m+1}$. 

We consider only the $\ceil{\log_s(n^2)}+1$ important intervals, i.e., the intervals $[s^{i},s^{i+1})$,
$m-\ceil{\log_s(n^2)} \leq i \leq m$, and index them by $i$.
Let $k$ be the index of the interval of weights contributing the
most weight to $\opt_{\Pw}(\sigma)$.
The advice is a self-delimiting encoding of the index, $j$, of the first request with weight $w \in [s^{k},s^{k+1})$, plus the advice used by the given $c$-competitive
\P-algorithm.
This requires at most $b+O(\log (n))$ bits of advice. 

The algorithm rejects all requests before the $j$th. From the
$j$th request, the algorithm calculates the index $k$.
The algorithm accepts those requests which
would be accepted by the \P-algorithm when presented with the subsequence of $\sigma$ consisting of the requests with weights in $[s^{k},s^{k+1})$. Since, by exponential
sparsification, $\opt_{\Pw}$ accepts total 
weight at most $\frac{1}{n}\opt_{\Pw}(\sigma)$ from requests with
unimportant weights, and it accepts at least as much from
interval $k$ as from any of the other $\ceil{\log_s (n^2)}+1$
intervals considered, $\opt_{\Pw}$ accepts weight
at least $(1-\frac{1}{n})\frac{\opt_{\Pw}(\sigma)}{\ceil{\log_s (n^2)}+1}$
from interval $k$. 
The algorithm, \alg, described here accepts at least
$\frac{1}{c}$  as many requests as $\opt_{\Pw}$ does in this interval,
and each of the requests it accepts is at least a fraction $\frac{1}{s}$
as large as the largest weight in this interval.
Thus, $c(1+\varepsilon)\alg(\sigma) \geq \left( \frac{1-\frac{1}{n}}
{\ceil {\log_s (n^2)}+1}\right)\opt_{\Pw}(\sigma)$, so \alg is 
$O(c \log n)$-competitive.
\qed
\end{proof}

In the online matching problem, edges arrive one by one. Each request
contains the names of the edge's two endpoints (the set of
endpoints is not known from the beginning, but revealed gradually as
the edges arrive). The algorithm
must irrevocably accept or reject them as they arrive, and the
goal is to maximize the number of edges accepted. The natural greedy
algorithm for this problem is well known to be $2$-competitive. In
terms of advice, the problem is known to be in \aoc, but is not 
\aocc~\cite{BFKM15}. We remark that a version of unweighted online matching with vertex arrivals (incomparable to our weighted matching with edge arrivals) has been studied with advice in~\cite{DBLP:journals/corr/DurrKR16}.

\begin{corollary}
There exists a $O(\log n)$-competitive algorithm for Weighted Matching reading $O(\log n)$ bits of advice.
\end{corollary}
\begin{proof}
The result follows from Theorem~\ref{thm:WeightedAOCgeneral} since there exists a 
$2$-competitive algorithm without advice for (unweighted) Matching.
\qed
\end{proof}

\subsection{Lower bounds}
First, we present a result which holds for 
the weighted versions of many maximization problems in \aoc. 
It also holds for the
weighted versions of \aocc minimization problems, but 
Theorem~\ref{thm:reduction} gives a much stronger result.

\begin{theorem} \label{thm:simplelowerbound}
For the weighted online versions of Independent Set, Clique, Disjoint Path Allocation, and Matching, 
an algorithm reading $o(\log n)$ bits of advice cannot be $f(n)$-competitive for
 any function  $f$.
\end{theorem}

To prove Theorem~\ref{thm:simplelowerbound}, we start by proving the
following lemma from which the theorem easily follows.
\begin{lemma}
Let $\P\in\aoc$ and suppose there exists a family $(\sigma_n)_{n\in\mathbb{N}}$ of $\P$-inputs with the following properties:
\begin{enumerate}
\item $\sigma_n=\langle r_1, r_2,\ldots , r_n\rangle$ consists of $n$ requests.
\item $\sigma_{n+1}$ is obtained by adding a single request to the end of $\sigma_n$.
\item If \P is a maximization problem, the feasible solutions are those in which at most one request is accepted.

If \P is a minimization problem, the feasible solutions are those in
which at least one request is accepted.
\end{enumerate}
Then, 
no algorithm for the weighted problem $\Pw$ reading $o(\log n)$ bits of advice
can be $f(n)$-competitive for any function $f$.
\end{lemma}
\begin{proof}
Let $\ALG$ be a $\Pw$-algorithm reading at most $b=o(\log n)$ bits of advice. Let $f(n)>0$ be an arbitrary non-decreasing function of $n$. We will show that for all sufficiently large $n$, there exists an input of length $n$ such that the profit obtained by $\OPT$ is at least $f(n)$ times as large as the profit obtained by $\ALG$. Since $f(n)$ was arbitrary, it follows that $\ALG$ is not 
$f(n)$-competitive for any function $f$.

Since $b=o(\log n)$, there exists an $N \in \mathbb{Z}$ such that for any $n \geq N$, $\ALG$ reads less than $\log (n)-1$
bits of advice on inputs of length at most $n$. Fix an $n\geq N$. For $1\leq i\leq n$, define the
$\Pw$-input $\widehat{\sigma}_i=\langle\{r_1, f(n)\}, \{r_2,
f(n)^2\},\ldots ,$ $\{r_i, f(n)^i\}\rangle$. Consider the set of inputs
$\{\widehat{\sigma}_1,\ldots , \widehat{\sigma}_{n}\}$. For every
$1\leq i\leq n$, the number of advice bits read by $\ALG$ on the input
$\widehat{\sigma}_i$ is at most $\log(n)-1$ (since the length of the
input $\widehat{\sigma}_i$ is $i\leq n$). Thus, by the pigeonhole
principle, there must exist two integers $n_1, n_2$ with $n_1<n_2$
such that $\ALG$ reads the same advice on $\widehat{\sigma}_{n_1}$ and
$\widehat{\sigma}_{n_2}$. If $\ALG$ rejects all requests in
$\widehat{\sigma}_{n_1}$, then it achieves a profit of $0$ while \OPT
obtains a profit of $f(n)^{n_1}$. If $\ALG$ accepts a request in
$\widehat{\sigma}_{n_1}$, then it obtains a profit of at most
$f(n)^{n_1}$. Since $\ALG$ reads the same advice on
$\widehat{\sigma}_{n_1}$ and $\widehat{\sigma}_{n_2}$ and since the
two inputs are indistinguishable for the first $n_1$ requests, this
means that $\ALG$ also obtains a profit of at most $f(n)^{n_1}$ on the
input $\widehat{\sigma}_{n_2}$. But
$\OPT(\widehat{\sigma}_{n_2})=f(n)^{n_2}$, and hence
$\OPT(\widehat{\sigma}_{n_2})/\ALG(\widehat{\sigma}_{n_2})\geq
f(n)^{n_2-n_1}\geq f(n)$.

For minimization problems, we can use the same arguments and the
input sequence $\widehat{\sigma}_i=\langle\{r_1, f(n)^{-1}\}, \{r_2,
f(n)^{-2}\},\ldots , \{r_i, f(n)^{-n}\}\rangle$.
\qed
\end{proof}

\begin{proof}[Proof of Theorem \ref{thm:simplelowerbound}]
For Independent Set, we can use the above lemma with a family of cliques $(K_n)_{n\in\mathbb{N}}$, and for Clique, we can use a family of independent sets.
For Matching, we can use a family of stars $(K_{1,n})_{n\in\mathbb{N}}$. 
For Disjoint Path Allocation, we use a path $P_{2n} = \langle v_1, v_2, \ldots, v_{2n} \rangle$ and $r_i = \langle v_i, v_{i+1}, \ldots v_{i+n} \rangle$. 
\qed
\end{proof}

Returning to the example of Weighted Matching, we now know that
$O(\log n)$ bits suffice to be $O(\log n)$-competitive, and that
no algorithm can be $f(n)$-competitive for any function $f$ with
$o(\log n)$ bits of advice.
In order to prove that a linear number of advice bits is necessary to achieve
constant competitiveness for Weighted Matching, we use a direct product
theorem from~\cite{M16}. This uses the concept defined in~\cite{M16} 
of a problem being 
{\em $\Sigma$-repeatable}. Informally, this means that it is always possible 
to combine $r$ (sufficiently
profitable) input sequences $I_1,I_2,\ldots,I_r$ into a single input
$g(I_1,I_2,\ldots,I_r)$ such that serving this single input gives profit
close to that of serving each of the $I_i$ independently and adding
the profits.

\begin{definition}
Let $P$ be an online maximization problem and $I$ be the set of possible
input sequences. Assume that for every input in $I$, there are
only a finite number of valid outputs. Let $I^*$ be the set of concatenations of
sequences (rounds) from $I$. $P$ is \emph{$\Sigma$-repeatable with
parameters $(k_1,k_2,k_3)$} if there exists a function $g ~:~ I^*\rightarrow
I$ satisfying the following:
\begin{itemize}
\item  For every  $\sigma^* \in I^*$ with $r$ rounds,
 $|g(\sigma^*)| \leq |\sigma^*| +k_1 r$
\item For every deterministic algorithm \ALG for $P$, there is a
deterministic algorithm $\ALG^*$ for sequences from $I^*$, such that
for every  $\sigma^* \in I^*$ with $r$ rounds,
$\ALG^*(\sigma^*) \geq \ALG(g(\sigma^*)-k_2 r$,
\item Let $\OPT^*$ denote an optimal algorithm for sequences from $I^*$.
For every  $\sigma^* \in I^*$ with $r$ rounds, 
$\OPT^*(\sigma^*) \leq \OPT(g(\sigma^*))+k_3 r$.
\end{itemize}
\end{definition}
\begin{theorem}
\label{thm:matching}
An $O(1)$-competitive algorithm for Weighted Matching must read $\Omega(n)$ bits of advice.
\end{theorem}
\begin{proof}
We prove the lower bound using a direct product theorem~\cite{M16}. 
According to~\cite{M16}, 
it suffices to show that: (i) Weighted Matching is $\Sigma$-repeatable, and
(ii) for every $c$, there exists a probability distribution $p_c$ with finite 
support such that for every deterministic algorithm $\D$ without advice, 
it holds that $\E_{p_c}[\OPT(\sigma)]\geq c\cdot \E_{p_c}[\D(\sigma)]$. 
Also, there must be a finite upper bound on the profit an algorithm can obtain on an input in the support of $p_c$.

It is trivial to see that Weighted Matching is $\Sigma$-repeatable. Fix $c\geq 1$ and let $k=2c-1$. We define the probability distribution $p_c$ by specifying a probabilistic adversary: The input graph will be a star $K_{1,m}$ consisting of $m$ edges for some $1\leq m\leq k$. In round $i$, the adversary reveals the edge $e_i=(v,v_i)$ where $v_i$ is a new vertex and $v$ is the center vertex of the star. The edge $e_i$ has weight $2^i$. If $i<k$, then with probability $1/2$ the adversary will proceed to round $i+1$, and with probability $1/2$ the input sequence will end. If the adversary reaches round $k$, it will always stop after revealing the edge $e_k$ of round $k$. Note that the support of $p_c$ and the largest profit an algorithm can obtain on any input in the support of $p_c$ are both finite.

Let $X$ be the random variable which denotes the number of edges revealed by the adversary. Note that $\Pr(X=j)=2^{-j}$ if $1\leq j<k$. Consequently, 
\begin{equation}
\Pr(X=k)=1-\Pr(X<k)=1-\sum_{i=1}^{k-1}2^{-i}=2^{-(k-1)}.
\end{equation}
Let $\D$ be a deterministic algorithm without advice. We may assume that $\D$ decides in advance on some $1\leq j\leq k$ and accepts the edge $e_j$ (the only other possible deterministic strategy it to never accept an edge, but this is always strictly worse than following any of the $k$ strategies that accepts an edge). If $X<j$, then the profit obtained by $\D$ is zero. If $X\geq j$, then $\D$ obtains a profit of $2^j$. It follows that 
\begin{align*}
\E[\D(\sigma)]=\Pr(X\geq j)2^j=(1-\Pr(X<j))2^j=2^{-(j-1)}2^j=2.
\end{align*}
The optimal algorithm $\OPT$ always accepts the last edge of the input. Thus, if $X=j$, then the profit of $\OPT$ is $2^j$. It follows that
\begin{align*}
\E[\OPT(\sigma)]=\sum_{j=1}^{k}\Pr(X=j)2^j=\sum_{j=1}^{k-1}\left(2^{-j}2^j\right)+2^{-(k-1)}2^k=k+1.
\end{align*}
Thus, we conclude that $\E[\OPT(\sigma)]\geq \frac{k+1}{2}\E[\D(\sigma)]= c\E[\D(\sigma)]$.
\qed
\end{proof}
 In particular, we cannot achieve constant competitiveness using $O(\log n)$ bits of advice for Weighted Matching. We leave it as an open problem to close the gap between $\omega(1)$ and $O(\log n)$ on the competitiveness of Weighted Matching algorithms with advice complexity $O(\log n)$.

\section{Scheduling with Sublinear Advice}
For the scheduling problems studied, the requests are jobs, each
characterized by its size.
Each job must be assigned to one of $m$ available machines.
If the machines are {\em identical}, the {\em load} of a job on any machine
is simply its size.
If the machines are {\em related}, each machine has a speed, and the
load of a job, $J$, assigned to a machine with speed $s$ is the size of
$J$ divided by $s$.
If the machines are {\em unrelated}, each job arrives with a vector
specifying its load on each machine.

Consider
a sequence $\sigma=\langle r_1,\ldots , r_n\rangle$ of $n$ jobs that arrive online. Each job $r_i\in\sigma$ has an associated weight-function $w_i:[m]\rightarrow\mathbb{R}_{+}$. Upon arrival, a job must irrevocably be assigned to one of the $m$ machines. The \emph{load} $L_j$ of a machine $j\in[m]$ is defined as $L_j=\sum_{i\in M_j}w_i(j)$ where $M_j$ is the set of (indices of) jobs scheduled on machine $j$. The \emph{total load} of a schedule for $\sigma$ is the vector $\vl=(L_1,\ldots , L_m)$. 
We say that $(L_1,\ldots , L_m) \leq (L'_1,\ldots , L'_m)$ if and only if
$L_i \leq L'_i$ for $1 \leq i \leq m$. 
A scheduling problem of the above type is specified by an \emph{objective function} $f:\mathbb{R}_+^m\rightarrow \mathbb{R}_+$ and by specifying if the goal is to minimize or maximize $f(\vl)=f(L_1,\ldots, L_m)\in\mathbb{R}_+$. 
We assume that $f$ is non-decreasing, i.e., $f(\vl)\leq f(\vl')$ for all $\vl\leq \vl'$. 
Some of the classical choices of objective function include: 
\begin{itemize}
\item Minimizing the $\ell_p$-norm $f_p(\vl)=f_p(L_1,\ldots , L_m)=\|(L_1,\ldots , L_m)\|_p$ for some $1\leq p\leq \infty$. That is, for $1\leq p<\infty$, the goal is to minimize $\left(\sum_{j\in[m]} L_j^p\right)^{1/p}$ and for $p=\infty$, the goal is to minimize the makespan $\max_{j\in [m]} L_j$.
\item Maximizing the minimum load $f(\vl)=\min_{j\in [m]} L_j$. 
This is also known as machine covering.
Note that this objective function is not a norm\footnote{$f$ is a norm if $f(\alpha \vv)=|\alpha|f(\vv)$, $f(\vu+\vv)\leq f(\vu)+f(\vv)$, and $f(\vv)=0 \Rightarrow \vv=\mathbf{0}$.}, but it does satisfy that $f(\alpha \vl)=\alpha f(\vl)$ for every $\alpha\geq 0$ and $\vl\in\mathbb{R}^m_+$.
\end{itemize}
We begin with a result for \emph{unrelated} machines. 

\begin{theorem}
\label{thm:unrelatedmin}
Let \P be a scheduling problem on $m$ unrelated machines where the goal is to minimize an objective function $f$. 
Assume that $f$ is a non-decreasing norm. Then, for $0<\varepsilon\leq 1$, there exists a $(1+\varepsilon)$-competitive $\P$-algorithm reading $O\big( (\frac{4}{\varepsilon} \log(n)+2)^m \log (n)\big)$ bits of advice. 
In particular, if $m=O(1)$ and $\varepsilon=\Omega(1)$, then there exists a $(1+\varepsilon)$-competitive algorithm reading $O(\polylog(n))$ bits of advice.
\end{theorem}
\begin{proof}
Since the objective function $f$ is a norm on $\mathbb{R}^m$, we will
denote it by $\|\cdot \|$.
Let $\mathbf{1}_j$ be the $j$th unit vector
(the vector with a $1$ in the $j$th coordinate and $0$ elsewhere). 

Fix an input sequence $\sigma$. 
The oracle starts by computing an arbitrary optimal schedule for $\sigma$. 
Throughout most of this
proof, we assume that $n$ is sufficiently large. The necessary conditions
are discussed at the end of the proof, along with how to handle small $n$.

Let $\vl_{\OPT}$ be the load-vector of this schedule. Thus,
$\OPT(\sigma)=\|\vl_\OPT\|$.
Let $s=1+\varepsilon /2$ and let $k$ be the unique integer such that
$s^{k}\leq \|\vl_\OPT\| < s^{k+1}$. 
A job $r_i\in\sigma$ is said to be \emph{unimportant}  
if there exists a machine $j\in[m]$ such that $\|w_i(j)\mathbf{1}_j\|<
s^{k-\ceil{\log_{s}(n^2)}}$.
A job which is not unimportant is {\em important}. 

The oracle uses $O(\log n)$ bits to encode $n$ using a self-delimiting
encoding. 
It then writes the index $i'$ of the first important job $r_{i'}$ onto
the advice tape 
(or indicates that $\sigma$ contains no important jobs) using $\ceil{\log(n+1)}$ bits. Let $j'$ be the machine minimizing $\|w_{i'}(j')\mathbf{1}_{j'}\|$,
 where ties are broken arbitrarily. 
 The oracle also writes $\Delta_{i'} = k-k_{i'}$, where $k_{i'}$ is the unique integer 
 such that $s^{k_{i'}}\leq \|w_{i'}(j')\mathbf{1}_{j'}\|< s^{k_{i'}+1}$ onto the advice tape using $O(\log \log_s n)$ bits.

\emph{Scheduling unimportant jobs.} If a job $r_i\in\sigma$ is unimportant, then the algorithm schedules the job on the machine $j$ minimizing $\|w_i(j)\mathbf{1}_j\|$ where ties are broken arbitrarily. We now explain how the algorithm knows if a job is unimportant or not. If $r_i\in\sigma$ is a job that arrives before the first important job, i.e., if $i< i'$, then $r_i$ is unimportant by definition. When job $r_{i'}$ arrives, the algorithm can deduce $k$ since it knows $\Delta_{i'}$ from the advice and since it can compute $\min_{j} \|w_{i'}(j)\mathbf{1}_{j}\|$ without help. Knowing $k$ (and the number of jobs $n$), the algorithm is able to tell if a job is unimportant or not.

\emph{Scheduling important jobs.}
We now describe how the algorithm schedules the important jobs. To this end, we define the type of an important job. For an important job $r_i$, let
 $\Delta_i(1),\ldots , \Delta_i(m)$ be defined as follows: 
 For $1 \leq j \leq m$, if there exists an integer
 $k_i(j) \leq k$ such that $s^{k_i(j)}\leq \|w_{i}(j)\mathbf{1}_{j}\| <s^{k_i(j)+1}$, 
 then $\Delta_i(j)=k-k_i(j)$ (since $r_i$ is important, $\Delta_i(j) \leq \ceil{\log_{s}(n^2)}$). 
 If no such integer exists, then it must be the case that $ \|w_{i}(j)\mathbf{1}_{j}\|\geq s^{k+1}> \|\vl_\OPT\|$. In this case, we let $\Delta_i(j)=\bot$ be a dummy symbol. The \emph{type} of $r_i$ is the vector $\mathbf{\Delta}_i=(\Delta_i(1),\ldots ,\Delta_i(m))$. Note that there are only $(\lceil \log_{s}(n^2)\rceil+2)^m$ different types. For each possible type $\mathbf{\Delta}=(\Delta(1),\ldots , \Delta(m))$, the oracle writes the number, $a_{\mathbf{\Delta}}$, of jobs of type $\mathbf{\Delta}$ onto the advice tape. This requires at most $(\lceil \log_{s}(n^2)\rceil+2)^m\lceil \log(n+1)\rceil$ bits of advice.

Note that since $\|\cdot\|$ is a norm, if $r_i\in\sigma$ is of type $\mathbf{\Delta}_i=(\Delta_i(1),\ldots , \Delta_i(m))$, then $s^{k-\Delta_i(j)}\|\mathbf{1}_j\|^{-1}\leq w_i(j)\leq s^{k-\Delta_i(j)+1}\|\mathbf{1}_j\|^{-1}$ if $\Delta_i(j)\neq\bot$ and $w_i(j)>\|\vl_\OPT\|\|\mathbf{1}_j\|^{-1}$ if $\Delta_i(j)=\bot$. The algorithm computes an optimal schedule $\widehat{S}_{\text{\imp}}$ for the input $\widehat{\sigma}$ which for each possible type $\mathbf{\Delta}$ contains $a_{\mathbf{\Delta}}$ jobs with weight-function $\widehat{w}_{\mathbf{\Delta}}$ where $\widehat{w}_{\mathbf{\Delta}}(j)=s^{k-\Delta(j)+1}\|\mathbf{1}_j\|^{-1}$ if $k(j)\neq\bot$ and $\widehat{w}_{\mathbf{\Delta}}(j)=\infty$ otherwise. This choice of weight-function ensures that if $r_{i}\in\sigma$ is a job of type $\mathbf{\Delta}_i$, then for each $j$ with $\Delta_i(j)\neq\bot$, 
\begin{equation}
\label{eq:unrelated}
w_{i}(j)< \widehat{w}_{\mathbf{\Delta}_i}(j)\leq s \cdot w_{i}(j).
\end{equation}
When an important job of $\sigma$ arrives, the algorithm computes the
type of the job. Based solely on this type, the algorithm schedules
the important jobs in $\sigma$ by following the schedule
$\widehat{S}_{\imp}$ for $\widehat{\sigma}$. Let $\vl_{\imp}$ be the
load-vector of the important jobs of $\sigma$ scheduled by
$\ALG$. Note that by Eq.~(\ref{eq:unrelated}), the weight-function of
an important job of $\sigma$ is strictly smaller (for all machines)
than the weight-function of the corresponding job of
$\widehat{\sigma}$. Thus, since $f$ is non-decreasing $\|\vl_{\imp}\|$ is bounded from above by the cost of the schedule $\widehat{S}_{\imp}$ for $\widehat{\sigma}$.

\emph{Putting it all together.} The optimal schedule for $\sigma$
computed by the oracle induces a schedule of $\widehat{\sigma}$. Let $\widehat{\vl}$ be the load-vector of this schedule. By Eq.~(\ref{eq:unrelated}), we get that $\|\widehat{\vl}\| \leq s \|\vl_\OPT\|$. Thus, the cost of $\widehat{S}_{\text{\imp}}$ (which was an optimal scheduling of $\widehat{\sigma}$) is at most $\|\widehat{\vl}\|\leq s\|\vl_\OPT\|$.

Let $\vl_{\unimp}$ be the load-vector of the unimportant jobs scheduled by $\ALG$. 
Furthermore, let $M_j$ be the set of indices of the unimportant jobs scheduled by $\ALG$ on machine $j$. By subadditivity, 
\begin{align*}
\|\vl_{\unimp}\|&=\left\|\sum_{j=1}^m\sum_{i\in M_j}w_i(j)\mathbf{1}_j\right\| \leq \sum_{j=1}^m\sum_{i\in M_j} \|w_i(j)\mathbf{1}_j\| < \sum_{j=1}^m\sum_{i\in M_j} s^{k-\lceil \log_{s}(n^2)\rceil} \\
&\leq n \frac{\|\vl_\OPT\|}{n^2}\leq \frac{\|\vl_\OPT\|}{n}. 
\end{align*}
We are finally able to bound the cost of the entire schedule created by $\ALG$:
\begin{align*}
\ALG(\sigma)=\|\vl_{\imp}+\vl_{\unimp}\|\leq \|\vl_{\imp}\|+\|\vl_{\unimp}\|\leq (s+1/n)\|\vl_\OPT\|
\end{align*}
Recall that $s=1+\varepsilon / 2$. Thus, if $n \geq 2/\varepsilon$, then $\ALG(\sigma)\leq (1+\varepsilon)\OPT(\sigma)$. For inputs of length less than $2/\varepsilon$, the oracle can simply encode the optimal solution using at most $\frac{2}{\varepsilon}\lceil \log m\rceil$ bits of advice. The total amount of advice used by our algorithm is at most
\begin{align*}
(\lceil \log_{s}(n^2)\rceil+2)^m\lceil \log(n+1)\rceil+O(\log n+\log \log_s n)=O\big( (4\varepsilon^{-1} \log(n)+2)^m \log (n)\big).
\end{align*}
\qed
\end{proof}

For the following discussion, assume that $\varepsilon=\Theta(1)$. We remark that the $(1+\varepsilon)$-competitive algorithm in Theorem~\ref{thm:unrelatedmin} is only of interest if the number of machines $m$ is small compared to the number of jobs $n$. As already noted, the most interesting aspect of Theorem~\ref{thm:unrelatedmin} is that our algorithm uses only $\polylog(n)$ bits of advice if $m$ is a constant. More generally, if $m=o(\log n / \log\log n)$, then our algorithm will use $o(n)$ bits of advice. On the other hand, if $m=\Theta(\log n)$, then our algorithm uses $\Omega(\log(n)^{\log (n)})$ bits of advice, which is worse than the trivial $1$-competitive algorithm which uses $n\lceil \log m\rceil=O(n\log\log n )$ bits of advice when $m=\Theta(\log n)$.

The advice complexity of the algorithm in Theorem~\ref{thm:unrelatedmin} depends on the number of machines $m$ because we want the result to hold even when the machines are unrelated. 
We now show that when restricting to related machines, we can obtain a $(1+\varepsilon)$-competitive algorithm using $O(\varepsilon^{-1}\log^2 n)$ bits of advice, independent of the number of machines.
The proof resembles that of Theorem~\ref{thm:unrelatedmin}.
The main difference is that we are able to reduce the number of types to $O(\log^2 n)$.

\begin{theorem}\label{thm:relatedmin}
Let \P be a scheduling problem on $m$ related machines where the goal is to minimize an objective function $f$. Assume that $f$ is a non-decreasing norm. Then, for $0<\varepsilon\leq 1$, there exists a $(1+\varepsilon)$-competitive $\P$-algorithm with advice complexity $O\big(\varepsilon^{-1}\log^2 n\big).$
\end{theorem}
\begin{proof}
Since the objective function $f$ is a norm on $\mathbb{R}^m$, we will denote it by $\|\cdot \|$. Fix an input sequence $\sigma$. The oracle starts by computing an arbitrary optimal schedule for $\sigma$. Let $s=1+\varepsilon/2$. The oracle uses $O(\log n)$ bits to encode $n$ using a self-delimiting encoding. 

Let $C_1,\ldots , C_m$ be the speeds of the $m$ machines. 
Assume without loss of generality that $\|\mathbf{1}_j\|/C_j$ attains its minimum value when $j=1$. 
Define $B=\|\mathbf{1}_1\|/C_1$. Let $\vl_{\OPT}$ be the load-vector of the fixed optimal schedule. 
Thus, $\OPT(\sigma)=\|\vl_\OPT\|$. Let $k$ be the unique integer such that $s^{k}\leq \|\vl_\OPT\| < s^{k+1}$. 
A job $r_i\in\sigma$ is said to be \emph{unimportant} if its weight,
$w_i$, satisfies $w_iB < s^{k-\ceil{\log_{s}(n^2)}}$. 
A job which is not unimportant is important. 
Note that $w_iB$ is always bounded from above by $\|\vl_\OPT\|$ since $r_i$ must be placed on some machine. 
The oracle writes the index $i'$ of the first important job $r_{i'}$ onto the advice tape 
(or indicates that $\sigma$ contains no important jobs) using $\ceil{\log n}+1$ bits. 
The oracle also writes the unique integer $k'$ such that $s^{k-k'}\leq w_{i'}B < s^{k-k'+1}$ onto the advice tape, using $O(\log\log_{s}(n))$ bits.

We now explain how the algorithm knows if a job is unimportant or not. If $r_i\in\sigma$ is a job that arrives before the first important job, i.e., if $i< i'$, then $r_i$ is unimportant by definition. When job $r_{i'}$ arrives, the algorithm can deduce $k$ since it knows $k'$ from the advice and since it can compute $w_{i'}B$ without help. Knowing $k$ (and the number of jobs $n$), the algorithm is able to tell if a job is unimportant or not.

If a job $r_i\in\sigma$ is unimportant, then the algorithm schedules the job on machine $1$.

\emph{Scheduling important jobs.}
We  now describe how the algorithm schedules the important jobs. 
To this end, we define the type of an important job. 
The \emph{type} of an important job $r_i$ is the non-negative integer $t_i$ such that $s^{k-t_i}\leq w_iB< s^{k-t_i+1}$. Note that there are only $\ceil{\log_{s}(n^2)}+1$ different types. For each possible type $0\leq t\leq \ceil{\log_{s}(n^2)}$, the oracle writes the number of jobs $a_t$ of that type onto the advice tape. This requires at most $O(\log_{s}(n^2)\log(n))$ bits of advice.

Note that since $\|\cdot\|$ is a norm, if $r_i\in\sigma$ is of type $t_i$, then $s^{k-t_i}B^{-1} \leq w_i\leq s^{k-t_i+1}B^{-1}$. 
The algorithm computes an optimal schedule $\widehat{S}_{\text{\imp}}$ for the input $\widehat{\sigma}$ which for each possible type $0\leq t\leq \ceil{\log_{s}(n^2)}$ contains $a_{t}$ jobs with weight $\widehat{w}_t=s^{k-k_i+1}B^{-1}$. This choice of weight ensures that if $r_{i}\in\sigma$ is a job of type $t_{i}$, then, 
\begin{equation}
\label{eq:related}
w_{i}< \widehat{w}_{t_i}\leq s\cdot w_{i}.
\end{equation}
When an important job of $\sigma$ arrives, the algorithm computes the type of the job. Based solely on this type, the algorithm schedules the important jobs in $\sigma$ by following the schedule $\widehat{S}_{\imp}$ for $\widehat{\sigma}$. Let $\vl_{\imp}$ be the load-vector of the important jobs of $\sigma$ scheduled by $\ALG$. Note that by Eq.~(\ref{eq:related}), the weight of an important job of $\sigma$ is strictly smaller than the weight of the corresponding job of $\widehat{\sigma}$. Thus, $\|\vl_{\imp}\|$ is bounded from above by the cost of the schedule $\widehat{S}_{\imp}$ for $\widehat{\sigma}$.

\emph{Putting it all together.} The fixed optimal schedule for $\sigma$ induces a scheduling of $\widehat{\sigma}$. Let $\widehat{\vl}$ be the load-vector of this schedule. By Eq.~(\ref{eq:related}), we get that $\widehat{\vl}\leq s\vl_\OPT$. Thus, the cost of $\widehat{S}_{\text{\imp}}$ (which was an optimal scheduling of $\widehat{\sigma}$) is at most $\|\widehat{\vl}\|\leq s\|\vl_\OPT\|$.

Let $W_u$ be the total weight of unimportant jobs scheduled on machine $1$ by $\ALG$. We have that
\begin{align*}
\|(W_u/C_1)\mathbf{1}_1\|=W_uB\leq n s^{k-\ceil{\log_{s}(n^2)}}\leq n\frac{\|\vl_\OPT\|}{n^2}\leq \frac{\|\vl_\OPT\|}{n}.
\end{align*}
We are finally able to bound the cost of the entire schedule $S_{\imp}\cup S_{\unimp}$ created by $\ALG$: \begin{align*}
\ALG(\sigma)=\|\vl_{\imp}+(W_u/C_1)\mathbf{1}_1\|\leq \|\vl_{\imp}\|+\|(W_u/C_1)\mathbf{1}_1\|\leq (s+1/n)\|\vl_\OPT\|
\end{align*}
Recall that $s=1+\varepsilon / 2$. Thus, if $n>2/\varepsilon$, then $\ALG(\sigma)\leq (1+\varepsilon)\OPT(\sigma)$. For inputs of length less than $2/\varepsilon$, the oracle can simply encode the optimal solution using at most $\frac{2}{\varepsilon}\lceil \log m\rceil$ bits of advice. The total amount of advice used by our algorithm is $O(\varepsilon^{-1}\log^2 n)$.
\qed
\end{proof}

We now consider scheduling problems where the goal is to maximize an objective function $f$. Recall that we assume that the objective function is non-decreasing. The most notable example is when $f$ is the minimum load. In the following theorem, we show how to schedule almost optimally on unrelated machines with only a rather weak constraint on $f$ (weaker than $f$ being a norm).
\begin{theorem}\label{thm:seminorm}
Let $\P$ be a scheduling problem on $m$ unrelated machines where the goal is to maximize an objective function $f$.
 Assume that $f$ is non-decreasing, that $f(\alpha\vl) \leq \alpha f(\vl)$ for every $\alpha\geq 0$, and $\vl\in\mathbb{R}^m_+$.
 Then, for every $0<\varepsilon\leq 1$, there exists a $(1+\varepsilon)$-competitive $\P$-algorithm with advice complexity $O((\frac{4}{\varepsilon} \log(n)+2)^mm^2\log n).$
 In particular, if $m=O(1)$ and $\eps = \Omega(1)$, the advice complexity is $O(\polylog (n))$.
\end{theorem}
\begin{proof}
Fix an input sequence $\sigma$ and an arbitrary optimal schedule. Let $s=1+\varepsilon/2$. The oracle uses $O(\log n)$ bits to encode $n$ using a self-delimiting encoding.

For $1\leq j\leq m$, let $L_j$ be the load on machine $j$ in the
optimal schedule. Furthermore, let $k_j$ be the unique integer such
that $s^{k_j}\leq L_j<s^{k_j+1}$. 
We say that a job $r_i$ is \emph{unimportant to machine $j$} if
$w_i(j)<s^{k_j-\lceil \log_{s}(n^2)\rceil}$, \emph{important to
  machine $j$} if $s^{k_j-\lceil \log_{s}(n^2)\rceil}\leq
w_i(j)<s^{k_j+1}$ and {\em huge to machine $j$} if $w_i(j) \geq
s^{k_j+1}$.
Note that if $r_i$ is huge to machine $j$, \opt does not schedule
$r_i$ on machine $j$. A job which is important to
at least one machine is called {\em important}.
All other jobs are called {\em unimportant}.
Note that, by definition, any unimportant job is unimportant (and not huge) to the machine where it
is scheduled by \opt.
We number the machines such that the first job which is important to machine $j$ arrives no later than the first job which is important to machine $j'$ for every $j<j'$.
This numbering is written to the advice tape, using $O(m \log m)$ advice bits.

The algorithm works in $m+1$ phases (some of which might be empty). Phase 0 begins when the first request arrives. For $1\leq j\leq m$, phase $j-1$ ends and phase $j$ begins when the first important job for machine $j$ arrives. Note that the same job could be the first important job for more than one machine. Phase $m$ ends with the last request of $\sigma$.  
For each phase, $j$, the oracle writes the index, $i$, of the request starting the phase and the unique integer $\Delta_i(j)$ such that $s^{k_j-\Delta_i(j)} \leq w_i(j) < s^{k_j-\Delta_i(j)+1}$.

The unimportant jobs are scheduled arbitrarily by our algorithm (it will become clear from the analysis of the algorithm that any choice will do).

We now describe how the algorithm schedules the important jobs in phase $j$ for $1\leq j\leq m$. By definition, at any point in phase $j$, we have received an important job for machines $1,2,\ldots , j$ and no important job for machine $j+1$ has yet arrived. 

The \emph{type} of a job $r_i$ in phase $j$ is a vector $\mathbf{\Delta}_i=(\Delta_i(1),\ldots , \Delta_i(j))$ where $\Delta_i(j')$ is the interval of $r_i$ on machine $j'$
(so $s^{k_j-\Delta_i(j')} \leq w_i(j') < s^{k_j-\Delta_i(j')+1}$) or $\bot$ if the job is not important to machine $j'$. Note that there are $\ceil{(2+\log_{s}(n^2))^j}$ possible job types in phase $j$. The oracle considers how the jobs in phase $j$ are scheduled in the fixed optimal schedule. For each job type $\mathbf{\Delta}$ and each machine $1\leq j'\leq j$, the oracle encodes the number of jobs of that type which are scheduled on machine $j'$ during phase $j$. This can be done using $O(\lceil 2+\log_{s}(n^2)\rceil^mm\log n)$ bits of advice for a single phase, and  $O(\lceil 2+\log_{s}(n^2)\rceil^mm^2\log n)$ bits of advice for all $m$ phases.

Equipped with the advice described above, the algorithm simply schedules the important jobs in the current phase based on their types. This ensures that, for each machine $j$, the total load of important jobs that \opt schedules on machine $j$ is at most $s$ times as large as the total load of important jobs scheduled by \alg on machine $j$ (since if $r_i$ and $r_{i'}$ are important to machine $j$ and of the same type, then $w_i(j)< s\cdot w_{i'}(j)$).

In order to finish the proof, we need to show that the contribution of
unimportant jobs to $\OPT(\sigma)$ is negligible (recall that all jobs
are either important or unimportant). To this end, let $L_j^{\unimp}$ (resp. $L_j^{\imp}$) be the load on machine $j$ of the unimportant (resp. important) jobs scheduled on that machine in the optimal schedule. Note that $L_j=L_j^{\unimp}+L_j^{\imp}$. By the definition of an unimportant job (and since there trivially can be no more than $n$ unimportant jobs), we find that for every $1\leq j\leq m$,

\begin{equation*}
L_j^{\unimp}<n \cdot s^{-\ceil{\log_s(n^2)}} \cdot L_j\leq \frac{L_j}{n}.
\end{equation*}
Thus, $L_j=L_j^{\unimp}+L_j^{\imp}\leq L_j/n+L_j^{\imp}$ from which $L_j\leq \frac{n}{n-1}\cdot L_j^{\imp}$ follows, assuming that $n>1$. Since this holds for all machines, and since as previously argued $\vl_{\OPT}^{\imp}\leq s\cdot \vl_{\ALG}^{\imp} 
\leq s\cdot \vl_{\ALG}$, we get that

$$\vl_{\OPT}\leq \frac{n}{n-1} \vl_{\OPT}^{\imp}\leq s\frac{n}{n-1} \vl_{\ALG}.$$
By assumption, the objective function $f$ satisfies $f(\alpha \vl)\leq \alpha f(\vl)$ and is non-decreasing. Thus, we conclude that
\begin{align*}
\OPT(\sigma)&=f(\vl_{\OPT}) \leq f\left( s\frac{n}{n-1}\vl_{\ALG}\right)\leq s\frac{n}{n-1}f(\vl_{\ALG})=s\frac{n}{n-1}\ALG(\sigma).
\end{align*}
For $n \geq 2 + \frac{2}{\eps}$, this gives a ratio of at most $1+\eps$.
\qed
\end{proof}

\bibliographystyle{plain}
\bibliography{refs}
\clearpage

\appendix
\normalsize

\setcounter{theorem}{1}

\section*{Appendix}
\renewcommand{\thesubsection}{\Alph{subsection}}

\subsection{AOC-Complete Problems}
For completeness, we state the full definition of \minSk from~\cite{BFKM15}:
 \begin{definition}[\!\cite{BFKM15}]
\label{minSkdef}
The \emph{minimum asymmetric string guessing problem with known history}, \minSk
, has input $\langle ?,x_1,\ldots , x_n\rangle$, where $x = x_1 \ldots x_n \in\{
0,1\}^n$, for some $n\in\mathbb{N}$. For $1\leq i \leq n$, round $i$ proceeds as
 follows:
\begin{enumerate}
\item If $i>1$, the algorithm learns the correct answer, $x_{i-1}$, to the request in the previous round.  
\item The algorithm answers $y_i=f(x_1,\ldots , x_{i-1})\in\{0,1\}$, where $f$ is a function defined by the algorithm.

\end{enumerate}
The {\em output} $y=y_1\ldots y_n$ computed by the algorithm is \emph{feasible},
 if 
 $x\sq y$. Otherwise, $y$ is \emph{infeasible}. 
The \emph{cost} of a feasible output is $\no{y}$, and the cost of an infeasible 
output is $\infty$.
\end{definition}

In addition to \minasgk, the class of \aocc problems also contains many graph problems.
The following four graph problems are studied in the vertex-arrival
model, so the requests are vertices, each presented together with its
edges to previous vertices.
The first three problems are minimization problems and the last one is
a maximization problem.
In {\em  Vertex Cover}, an algorithm must accept a set of vertices
which constitute a vertex cover, so for every edge in the requested
graph, at least one of its endpoints is accepted. 
For {\em Dominating Set}, the
accepted vertices must constitute a dominating set, so every vertex
in the requested graph must be accepted, or one its neighbors
must be accepted.
In {\em  Cycle Finding}, an algorithm must accept a set of vertices
inducing a cyclic graph.
For {\em Independent Set}, the accepted vertices must form an
independent set, i.e., no two accepted vertices share an edge. 

For {\em Disjoint Path Allocation} a path $P$ is given, and the requests
are subpaths of $P$.
The aim is to accept as many edge disjoint paths as possible.

For {\em Set Cover}, the requests are finite subsets from a known
universe, and the union of the accepted subsets must be the entire universe.
The aim is to accept as few subsets as possible.

\subsection{Reductions for Theorem~\ref{thm:reduction}}
\label{reductions}
In the proof of Theorem~\ref{thm:reduction}, a reduction sketch was given
for the weighted online version of Vertex Cover. Here we include
sketches for the
reductions for the weighted versions of Cycle Finding, Dominating Set
and Set Cover.

\subsubsection{Cycle Finding}
Each input $\seq = \langle x_1, x_2, \ldots, x_n \rangle$ to the problem
\minasgk, is transformed to $f(\seq) = \langle v_1, v_2, \ldots, v_n \rangle$, where
$V=\{v_1, v_2, \ldots, v_n \}$ is the vertex set of a graph with edge
set $E=\{(v_j,v_i) \colon f'(x_i)=j\}\cup \{ (v_\textsc{min},
v_\textsc{max})\}$, where
$f'(x_i)$ is the largest $j<i$ such that $x_j=1$,
 $\textsc{max}$ is the largest $i$ such that $x_i=1$, and $\textsc{min}$ is
 the smallest $i$ such that $x_i=1$. If $\no{\seq}>2$, the vertices 
corresponding to $1$s form the only cycle in the graph.

The advice used by the \minasgk algorithm \ALGa consists of the advice used by 
the Cycle Finding
algorithm \ALGb in combination with 1 bit indicating whether or not $\no{\seq}\leq
2$ and in this case (an encoding of) one or two indices of $1$s in
the input sequence. If $\no{\seq}>2$, then
\ALGb accepts some vertices, and \ALGa returns a $1$ for the $x_i$
corresponding to each of those vertices. 

If \ALGa returns a non-optimal feasible
set, \ALGb does too, and the sets have the same weights, so
$\ALGa(\seq) \leq \ALGb(f(\seq))+\OPT(\seq)$. In this case, the
weights of the optimal solutions for $\seq$ and $f(\seq)$ are both the sum
of the weights of the elements corresponding to $1$s in $\seq$,
so $f$ is a length preserving $O(f(n))$-reduction.

\subsubsection{Dominating Set}
Each input $\seq = \langle x_1, x_2, \ldots, x_n \rangle$ to the problem
\minasgk, is transformed to $f(\seq) = \langle v_1, v_2, \ldots, v_n \rangle$, where
$V=\{v_1, v_2, \ldots, v_n \}$ is the vertex set of a graph with edge
set $E=\{(v_i,v_\textsc{max})\}.$, where
 $\textsc{max}$ is the largest $i$ such that $x_i=1$.

The advice used by the \minasgk algorithm \ALGa consists of the advice used by 
the Dominating Set
algorithm \ALGb in combination with 1 bit indicating whether or not $\no{\seq}=0$.
If $\no{\seq}\geq 1$, then there is another bit of advice indicating whether
or not \ALGb accepted $v_{\textsc{max}}$. If \ALGb did not accept
$v_{\textsc{max}}$, the advice also contains an index of a vertex corresponding
to a $0$ in $\seq$ which was accepted, plus the index of $v_{\textsc{max}}$.

If \ALGa's solution is
feasible, but not optimal, then $\no{\seq}>0$ and
\ALGb accepts some vertices, and \ALGa returns a $1$ for the $x_i$
corresponding to each of those vertices (though, in the case where
$v_{\textsc{max}}$ was rejected, it answers $1$ for $x_{\textsc{max}}$ and 
answers $0$ for the earlier request indicated by the advice).

If $\no{\seq}>0$ a minimum weight dominating
set for $f(\seq)$ consists of exactly those vertices corresponding to
$1$s in $\seq$, so the weights of the optimal solutions for $\seq$ and $f(\seq)$ 
are both the sum
of the weights of the elements corresponding to $1$s in $\seq$, unless
\ALGb did not accept $v_{\textsc{max}}$. However, the weight of $x_{\textsc{max}}
\leq \OPT(\seq)$, so
$\ALGa(\seq) \leq \ALGb(f(\seq))+\OPT(\seq)$.
Thus, $f$ is a length preserving $O(\log(n)$-reduction.

\subsubsection{Set Cover}
This reduction is very similar to that for Dominating Set.
Each input $\seq = \langle x_1, x_2, \ldots, x_n \rangle$ to the problem
\minasgk,  $\textsc{max}$ is the largest $i$ such that $x_i=1$.
In the set cover instance, the universe is
$\{ 1,\ldots,n\}$, and $f(\seq)$ is a set of $n$ requests, where
request $i$ is $\{ i\}$, unless $i=\textsc{max}$, in which case, the
set consists of $\textsc{max}$ and all of the $j$ where $x_j=0$.

As with the reduction to Dominating Set, this is a 
length preserving $O(\log(n)$-reduction.

\end{document}